\documentclass{jpp}
\usepackage{graphicx}
\usepackage{xcolor}
\usepackage{bm}
\usepackage{epstopdf, epsfig}
\usepackage{hyperref}
\usepackage{amsmath,amsfonts,amssymb}
\usepackage{subfigure}

\newtheorem{theorem}{Theorem}
\newtheorem{lemma}{Lemma}

\newtheorem{remark}{Remark}
\newtheorem{definition}{Definition}

\shorttitle{Admissibility}
\shortauthor{Burby, Kallinikos, MacKay, Perrella, and Pfefferl\'e}

\title{Characterization of admissible quasisymmetries}

\author{J. W. Burby\aff{1}\corresp{\email{joshua.burby@austin.utexas.edu}}, N. Kallinikos\aff{2}, R. S. MacKay\aff{2}, D. Perrella\aff{3}, and D. Pfefferl\'e\aff{3}}

\affiliation{\aff{1}Department of Physics and Institute for Fusion Studies, The University of Texas at Austin, Austin, TX 78712, USA
\aff{2}Mathematics Institute, University of Warwick, Coventry CV4 7AL, UK
\aff{3} The University of Western Australia, 35 Stirling Highway, Crawley, WA 6009, Australia}

\begin{document}

\maketitle

\begin{abstract}
We solve ``half" the problem of finding three-dimensional quasisymmetric magnetic fields that do not necessarily satisfy magnetohydrostatic force balance. This involves determining which hidden symmetries are admissible as quasisymmetries, and then showing explicitly how to construct quasisymmetric magnetic fields given an admissible symmetry. The admissibility conditions take the form of a system of overdetermined nonlinear partial differential equations involving second derivatives of the symmetry's infinitesimal generator. 
\end{abstract}


%
%
%

\section{Introduction}
Quasisymmetry was introduced in \cite{Boozer_QS_1983} as a condition on magnetic fields satisfying force balance, $(\nabla\times\bm{B})\times\bm{B} = \nabla p $, that ensures guiding center trajectories enjoy a constant of motion analogous to canonical angular momentum. Subsequently, \cite{Burby_phase_2013}, and then \cite{BKM_2020}, showed how to remove force balance from the definition of quasisymmetry, leading to
\begin{definition}
Let $Q\subset \mathbb{R}^3$ be a bounded spatial domain. A nowhere-vanishing vector field $\bm{B}:Q\rightarrow\mathbb{R}^3$ is said to be  \textbf{quasisymmetric} when there is a non-zero divergence-free vector field $\bm{u}$ on $Q$ such that
\begin{align}
\nabla\times(\bm{u}\times\bm{B}) = 0,\quad (\nabla\times\bm{B})\times \bm{u} + \nabla(\bm{u}\cdot\bm{B}) = 0,\quad \nabla\cdot \bm{B} = 0.\label{QS}
\end{align}
The field $\bm{u}$ is known as the \textbf{infinitesimal generator} of the quasisymmetry.
\end{definition}
\noindent This notion of quasisymmetry is readily shown to be equivalent to existence of a $1$-parameter family of spatial symmetries for a standard truncation (see \cite{Littlejohn_1981}, \cite{Littlejohn_1982}, \cite{Littlejohn_1983}, \cite{Littlejohn_1984}) of Littlejohn's guiding center Lagrangian. It should be compared with the related notion of weak quasisymmetry, where the second equation in \eqref{QS} is replaced with $\bm{u}\cdot\nabla |\bm{B}|^2 = 0$, introduced in \cite{Rodriguez_2020} and placed in a more general context in \cite{BKM_2021}; weak quasisymmetry is the weakest condition on $\bm{B}$ that ensures guiding centers enjoy an adiabatic invariant, associated with an approximate spatial symmetry, analogous to canonical angular momentum, while quasisymmetry is the weakest condition that guarantees a particular truncation of the guiding center equations is integrable by way of a spatial symmetry. (Velocity-dependent symmetries are discussed in \cite{BKM_2021}.) Removing the force balance constraint is of practical concern because, in light of Grad's conjecture (see \cite{Grad_conj_1967,Grad_1985,Constantin_Drivas_Ginsberg_grad_2021}), the equation $(\nabla\times\bm{B})\times\bm{B} = \nabla p$ is likely not a correct description of real stellarator equilibria at small enough scales. 

Many prior discussions of quasisymmetry attempt to address existence of smooth quasisymmetric $\bm{B}$. For instance, see \cite{Garren_Boozer_1991}, \cite{Landreman_Sengupta_2019}, \cite{BKM_2020}, \cite{Constantin_Drivas_Ginsberg_2021}, \cite{Landreman_Paul_2022}, \cite{Wechsung_2022}, \cite{Sato_2022}. However, as the above definition plainly demonstrates, the magnetic field $\bm{B}$ and the infinitesimal generator $\bm{u}$, i.e. the symmetry itself, play almost symmetric roles in the theory. Briefly thinking along these lines leads to the notion dual to quasisymmetry,
\begin{definition}
A divergence-free vector field $\bm{u}$ is \textbf{admissible} when there is some   quasisymmetric $\bm{B}$ with $\bm{u}$ as its infinitesimal generator.
\end{definition}
\noindent This Article initiates the study of admissible $\bm{u}$.

Characterization of admissible infinitesimal generators $\bm{u}$ plays at least two important roles in the theory of quasisymmetric magnetic fields. The first role concerns the usual existence problem. The process of finding quasisymmetric $\bm{B}$ may be split into two subprocesses: (1) find an admissible $\bm{u}$, then (2) find a quasisymmetric $\bm{B}$ with that $\bm{u}$ as its infinitesimal generator. Of course this decomposition will only be useful if there is some way to find admissible $\bm{u}$ without finding $\bm{u}$ and $\bm{B}$ simultaneously. Determining whether this can be done requires characterizing admissible $\bm{u}$. The second role concerns an important no-go result due to Garren-Boozer in \cite{Garren_Boozer_1991}. Garren-Boozer argued quasisymmetric $\bm{B}$ in toroidal $Q$ that satisfy force balance cannot exist, except when the $\bm{u}$-flow is a $1$-parameter family of Euclidean isometries. Their result has not been shown to apply when the force balance constraint is removed. In fact, three-dimensional \emph{weakly} quasisymmetric fields out of force balance are derived in \cite{Sato_2022}.  It is therefore very interesting to determine whether quasisymmetry must always correspond to rigid Euclidean motions, or if there are genuine three-dimensional examples of quasisymmetric magnetic fields. Characterization of admissible $\bm{u}$ addresses this problem directly.

An infinitesimal generator $\bm{u}$ on $Q$ is \textbf{Killing} when its flow is a $1$-parameter family of Euclidean isometries. Alternatively, the Killing condition can be characterized infinitesimally.
\begin{lemma}
A vector field $\bm{u}$ is Killing if and only if the \textbf{strain-rate tensor} $\bm{S} = (\nabla\bm{u}+\nabla\bm{u}^T)/2$ vanishes.
\end{lemma}
\begin{proof}
Let $\Phi_t:\bm{x}\mapsto \bm{x}_t$ denote the time-$t$ flow for $\bm{u}$, so that $\partial_t\Phi_t = \bm{u}\circ\Phi_t$. Let $\bm{v},\bm{w}\in\mathbb{R}^3$ be a pair of vectors. The time derivative of the dot product $\ell_t^2  = (\bm{v}\cdot\nabla\Phi_t)\cdot(\bm{w}\cdot\nabla\Phi_t)$ is 
\begin{align*}
\frac{d}{dt}\ell_t^2& = ([\bm{v}\cdot\nabla\Phi_t]\cdot \nabla\bm{u}\circ\Phi_t)\cdot (\bm{w}\cdot\nabla\Phi_t) + (\bm{v}\cdot\nabla\Phi_t)\cdot ([\bm{w}\cdot\nabla\Phi_t]\cdot\nabla\bm{u}\circ\Phi_t) \\
& = 2[\bm{v}\cdot\nabla\Phi_t]\cdot (\bm{S}\circ\Phi_t)\cdot [\bm{w}\cdot\nabla\Phi_t].
\end{align*}
If $\Phi_t$ is an isometry for each $t$ then $\ell_t^2 = \ell_0^2$ for all $t$. In particular, $0=d\ell_t^2/dt\mid_{t=0}=2\bm{v}\cdot\bm{S}\cdot\bm{w}$, since $\Phi_0 = 1$. Since $\bm{v},\bm{w}$ are arbitrary, this implies $\bm{S}=0$. Conversely, if $\bm{S}=0$ then $d\ell_t^2/dt =0$, which implies $\ell_t^2 = \ell_0^2$. Since $\bm{v},\bm{w}$ are arbitrary, this implies that $\Phi_t$ is an isometry for each $t$.
\end{proof}
\noindent Every Killing $\bm{u}$ is admissible, as the following argument shows. 
\begin{lemma}\label{killing_case}
Suppose $\bm{u}$ is a nowhere-vanishing vector field on $Q$ with vanishing strain-rate tensor, $\bm{S} = 0$. Then all vector fields $\bm{B}$ of the form 
\begin{align}
\bm{B} = \frac{1}{|\bm{u}|^2}\bigg(\bm{u}\times\nabla\psi + C\,\bm{u}\bigg),\label{B_stream_function_rep}
\end{align}
with $\bm{u}$-invariant functions $\psi,C$ (i.e. $\bm{u}\cdot\nabla \psi = \bm{u}\cdot \nabla C = 0$), are quasisymmetric with infinitesimal generator $\bm{u}$.
\end{lemma}
\begin{proof}
$\bm{S} = 0$ implies $\nabla\cdot \bm{u} = 0$ and $\nabla|\bm{u}|^2 + (\nabla\times\bm{u})\times\bm{u} = 0$. In particular we have $\bm{u}\cdot\nabla|\bm{u}|^2 = 0$.

Let $\bm{w} = \bm{u}\times\bm{B}$. We have
\begin{align*}
\bm{w} & = \frac{1}{|\bm{u}|^2}\bm{u}\times\bigg(\bm{u}\times\nabla\psi + C\,\bm{u}\bigg) = \frac{1}{|\bm{u}|^2}(\bm{u}\cdot\nabla\psi)\,\bm{u} - \nabla\psi = - \nabla\psi.
\end{align*}
This immediately gives us $\nabla\times(\bm{u}\times\bm{B}) = 0$.
Moreover, if $\mathcal{L}_{\bm{u}}$ denotes the linear operator on vector fields $\bm{v}$ given by $\mathcal{L}_{\bm{u}}\bm{v} = \bm{u}\cdot\nabla\bm{v} - \bm{v}\cdot\nabla\bm{u} = [\bm{u},\bm{v}]$ then
\begin{align*}
[\bm{u},\bm{B}] & =\bm{u}\cdot\nabla|\bm{u}|^{-2}(\bm{u}\times\nabla\psi + C\,\bm{u})\\
& + |\bm{u}|^{-2}\bigg(\mathcal{L}_{\bm{u}}\bm{u} \times\nabla\psi + \bm{u}\times \mathcal{L}_{\bm{u}}\nabla\psi + (\bm{u}\cdot\nabla C)\bm{u} + C\,\mathcal{L}_{\bm{u}}\bm{u}\bigg)= 0,
\end{align*} 
where we have used the fact that $\bm{u}$ is Killing to infer that $\mathcal{L}_{\bm{u}}$ is a cross-product derivation that commutes with $\nabla$, i.e. $\mathcal{L}_{\bm{u}}\nabla\psi = \nabla (\bm{u}\cdot\nabla\psi)$. Therefore 
\begin{align*}
0= \nabla\times\bm{w} = (\nabla\cdot\bm{B})\,\bm{u} - (\nabla\cdot\bm{u})\,\bm{B} -[\bm{u},\bm{B}] = (\nabla\cdot\bm{B})\,\bm{u},
\end{align*}
which implies $\nabla\cdot\bm{B} = 0$. 

It remains to show that the vector field $\bm{k} = (\nabla\times\bm{B})\times \bm{u} + \nabla(\bm{u}\cdot\bm{B})$ vanishes. Well-known vector identities imply
\begin{align*}
\bm{k} = [\bm{u},\bm{B}] + 2\,\bm{S}\cdot\bm{B}.
\end{align*}
But we have just shown that the first term vanishes, and the second term vanishes because $\bm{S} = 0$. So $\bm{B}$ is quasisymmetric with infinitesimal generator $\bm{u}$, as claimed.
\end{proof}
We remark that this Lemma can be generalised to fields for which $\psi$ is not a global function.  Given a $\bm{u}$-invariant function $C$ and an irrotational vector field $\bm{w}$ such that $\bm{u}\cdot\bm{w} =0$, define a field $\bm{B}$ by replacing $\nabla\psi$ with $\bm{w}$.  Then $\bm{B}$ is still divergence-free and quasisymmetric with infinitesimal generator $\bm{u}$.  This is the most general such $\bm{B}$. To prove this, define $\bm{w} = \bm{B}\times\bm{u}$. From $\nabla\times(\bm{u}\times\bm{B}) = 0$, this field is irrotational. Also $\bm{u}\cdot\bm{w} =0$, so $[\bm{u},\bm{w}] = 0$. Define $C = \bm{B}\cdot\bm{u}$. We have $\bm{u}\cdot\nabla C = 0$. Thus, $\bm{B} = (\bm{u}\times \bm{w} + C\bm{u})/|\bm{u}|^2$.

\noindent Every Killing field $\bm{u}$ is admissible because such functions $\psi,C$ exist (in abundance). For Euclidean metric, $\bm{u}$ is contained in the Euclidean Lie algebra, and therefore has the form $\bm{u}(\bm{x}) = \bm{U} + \bm{L} \times \bm{x}$ in Cartesian coordinates, for some constant vectors $(\bm{U},\bm{L})$.  By translation and rotation we can take $\bm{U},\bm{L}$ both in the $z$-direction.  If $\bm{U}=0$ (the case of axisymmetry) then we can choose $\psi,C$ to be any smooth functions of $r^2, z$ in cylindrical coordinates $(r,\phi,z)$.  In particular we can achieve $\nabla \psi \ne 0$ almost everywhere.  If $\bm{L}=0$ we can choose any smooth functions of $(x,y)$.  In the general case $\bm{U},\bm{L} \ne 0$  (helical symmetry), can choose any smooth functions of $r^2,\zeta$ with $\zeta = z - c\phi$ where $c\ne 0$ is the associated pitch (note this requires periodicity in $z$).

Thus, in the study of admissible $\bm{u}$, it is useful to restrict attention to infinitesimal generators with nowhere-vanishing strain-rate tensors. We call such fields $\bm{u}$ \textbf{non-Killing}. The existence of even one non-Killing admissible $\bm{u}$ would imply existence of genuinely-three-dimensional quasisymmetric magnetic fields. 
We leave out the case of $\bm{u}$ for which $\bm{S}=0$ in some places, nonzero in others.

This Article establishes three foundational results concerning admissible non-Killing $\bm{u}$. First, every such $\bm{u}$ must satisfy a system of partial differential equations that do not involve $\bm{B}$.
\begin{theorem}[necessary conditions for admissibility]\label{necessary_conditions}
If $\bm{u}$ is non-Killing and admissible then it must satisfy the PDE constraints
\begin{align}
0&=\text{\normalfont tr}(\bm{S}) \label{adm_nec_1}\\
0&=\text{\normalfont det}(\bm{S}) \label{adm_nec_2}\\
0&=\bm{u}\cdot \nabla\bm{\pi}_{\parallel} - \nabla\bm{u}^T\cdot \bm{\pi}_{\parallel} - \bm{\pi}_{\parallel}\cdot \nabla\bm{u} \label{adm_nec_3}
\end{align}
where $\bm{S} =\frac{1}{2} (\nabla\bm{u} + \nabla\bm{u}^T)$ and $\bm{\pi}_{\parallel} = 1 - \frac{2\bm{S}\cdot \bm{S}}{\bm{S} : \bm{S}}$ denotes the orthogonal projector onto $\text{ker}\,\bm{S}$.

\end{theorem}
\noindent See the discussion in Section \ref{discussion_sec} for an alternative formulation of Eq.\,\eqref{adm_nec_3}. Second, the necessary PDE constraints \eqref{adm_nec_1}-\eqref{adm_nec_3} on $\bm{u}$ are generally \emph{locally} sufficient for admissibility as well. 
\begin{theorem}[local admissibility]\label{local_admissibility_thm}
If $\bm{u}$ is non-Killing and satisfies the PDE constraints \eqref{adm_nec_1}-\eqref{adm_nec_3} then for each $\bm{x}\in Q$ such that
\begin{align*}
(1-\bm{\pi}_\parallel(\bm{x}))\cdot \bm{u}(\bm{x})\neq 0
\end{align*}
there is an open neighborhood $U\ni \bm{x}$ and quasisymmetric $\bm{B}$ defined within $U$ with $\bm{u}$ as its infinitesimal generator.
\end{theorem}
\noindent Finally, the local sufficiency condition extends to a global sufficiency condition in toroidal annuli as follows.
\begin{theorem}[global admissibility]\label{global_admissibility_thm}
Assume $Q$ is diffeomorphic to the toroidal annulus $S^1\times S^1\times [0,1]$, where $S^1 = \mathbb{R}/(2\pi\mathbb{Z})$ denotes the circle. Properties (I) and (II) for a vector field $\bm{u}$ are equivalent.
\begin{itemize}
\item[(I.a)] $\bm{u}$ is admissible and non-Killing with corresponding quasisymmetric magnetic field $\bm{B}:Q\rightarrow\mathbb{R}^3$.
\item[(I.b)] $\bm{u}\times \bm{B}$ is nowhere-vanishing in $Q$.
\item[(I.c)] Both $\bm{u}$ and $\bm{B}$ are tangent to $\partial Q$.
\end{itemize}

\begin{itemize}
\item[(II.a)] $\bm{u}$ is non-Killing and satisfies Eqs.\,\eqref{adm_nec_1}-\eqref{adm_nec_3}.
\item[(II.b)] $(1-\bm{\pi}_\parallel)\cdot \bm{u}$ is nowhere-vanishing in $Q$.
\item[(II.c)] Both $\bm{u}$ and $\text{\normalfont im}\,\bm{\pi}_\parallel$ are tangent to $\partial Q$.
\item[(II.d)] There is a smooth function $\psi:Q\rightarrow\mathbb{R}$ with nowhere-vanishing gradient such that $\bm{u}\cdot \nabla\psi = 0$ and $\bm{\pi}_\parallel\cdot\nabla\psi = 0$. In particular, $\psi$ is constant on $\partial Q$.
\end{itemize}

\end{theorem}

Properties (I.b), (I.c), (II.b), and (II.c) are technical conditions that roughly say $\bm{u}$ is either admissible or locally admissible in a ``non-degenerate way". Note that $(1-\bm{\pi}_\parallel)\cdot \bm{u}$ nowhere-vanishing is equivalent to the condition that $\bm{u}$ is never in the kernel of $\bm{S}$. Property (II.d) should be interpreted as a global admissibility conditions that restricts the topological behavior of solutions to Eqs.\,\eqref{adm_nec_1}-\eqref{adm_nec_3}. That property (II.d) is topological in nature can be understood as follows. The equation \eqref{adm_nec_3} is equivalent to $\mathcal{L}_{\bm{u}}(\bm{\pi}_\parallel) = 0$, where $\bm{\pi}_\parallel$ is regarded as a degree-$2$ symmetric contravariant tensor field. Note that $\bm{\pi}_\parallel$ is the orthogonal projector onto the line bundle $\text{ker}\,\bm{S}$. It follows that the rank-$2$ vector bundle $\text{span}(\bm{u})\oplus\text{ker}\,\bm{S}$ is integrable in the sense of Frobenius. Thus, in a neighborhood of any point in $Q$, there is a locally-defined function $\psi$ satisfying the properties in (II.d). It is unclear, however, if these locally-defined functions can be glued together to give a smooth globally-defined $\psi$. This is the topological question underlying (II.d). In the Appendix, we give an example of a pair of nowhere-vanishing vector fields $\bm{e}_3,\bm{u}$ defined on a toroidal annulus $Q$ with $|\bm{e}_3|=1$, $\nabla\cdot\bm{u}=0$, $[\bm{e}_3,\bm{u}]=0$, $\bm{e}_3\times\bm{u}$ nowhere-vanishing, and $\bm{e}_3,\bm{u}$ tangent to $\partial Q$, such that no global $\psi$ exists. This shows that if the local admissibility conditions \eqref{adm_nec_1}-\eqref{adm_nec_3} \emph{do} imply existence of a global $\psi$ by themselves then the reason must involve Eq.\,\eqref{adm_nec_2}. The appendix also contains a detailed proof that property (II.d) in Theorem \ref{global_admissibility_thm} can be replaced with the more obviously topological property
\begin{itemize}
\item[(II.d')] Each leaf of the foliation integrating $\text{span}(\bm{u})\oplus\text{ker}\,\bm{S}$ is compact.
\end{itemize}

\section{Derivation of the local admissibility constraints}\label{adm_deriv}
The following self-contained derivation of local admissibility constraints for $\bm{u}$ largely isolates elements of the discussion from \cite{BKM_2020}. That these necessary conditions are also sufficient, as will be explained in Sections \ref{local_adm_sec} and \ref{global_adm_sec}, was not recognized in \cite{BKM_2020}.

Suppose that $\bm{B}$ is quasisymmetric with non-Killing infinitesimal generator $\bm{u}$.  The following argument shows that $\bm{u}$ must satisfy the PDE constraints \eqref{adm_nec_1}-\eqref{adm_nec_3}. In addition, the discussion here leads to Lemma \ref{B_lemma}, which reduces the problem of finding a quasisymmetric $\bm{B}$ with given non-Killing $\bm{u}$ to a system of PDEs with a single unknown.

The equation system \eqref{QS} defining quasisymmetry is equivalent to
\begin{gather}
\bm{B}\cdot \nabla\bm{u} - \bm{u}\cdot \nabla\bm{B} = 0,\quad \bm{u}\cdot \nabla\bm{B} + \nabla\bm{u}\cdot \bm{B}=0,\quad \nabla\cdot\bm{B} = 0.\label{QS_alt}
\end{gather}
The well-known fact that $|\bm{B}|$ admits a symmetry for quasisymmetric $\bm{B}$ then follows from
\begin{align}
\bm{u}\cdot \nabla|\bm{B}|^2 &= 2\,(\bm{u}\cdot \nabla\bm{B})\cdot\bm{B} = 2\bm{B}\cdot\nabla\bm{u}\cdot \bm{B} = -2\bm{B}\cdot\nabla\bm{u}\cdot \bm{B} = 0.\label{modB_sym}
\end{align}

An immediate consequence of \eqref{QS_alt} is
\begin{align}
2\bm{B}\cdot \bm{S} = \bm{B}\cdot (\nabla\bm{u} + \nabla\bm{u}^T ) = \bm{B}\cdot \nabla\bm{u} + \nabla\bm{u}\cdot \bm{B} = 0.\label{kernel_condition}
\end{align}
Since, by definition, $\bm{B}$ is nowhere-vanishing, this implies $\text{det}(\bm{S}) = 0$, which recovers the constraint \eqref{adm_nec_2}. Clearly it applies to any admissible $\bm{u}$ whatsoever, be it Killing,    non-Killing, or anything in-between.

By symmetry of the strain-rate tensor, for each $\bm{x}\in Q$ there is an orthonormal frame $(\bm{e}_1(\bm{x}),\bm{e}_2(\bm{x}),\bm{e}_3(\bm{x}))$ that diagonalizes $\bm{S} (\bm{x})$: 
\[
\bm{S}(\bm{x}) = \lambda_1(\bm{x})\,\bm{e}_1(\bm{x})\bm{e}_1(\bm{x}) + \lambda_2(\bm{x})\,\bm{e}_2(\bm{x})\bm{e}_{2}(\bm{x}) + \lambda_3(\bm{x})\,\bm{e}_3(\bm{x})\bm{e}_3(\bm{x}).
\]
The constraint $\text{det}(\bm{S}) = 0$ implies that at least one of the eigenvalues, say $\lambda_{3}(\bm{x})$, is zero. The incompressibility constraint $0=\nabla\cdot \bm{u} = \text{tr}(\bm{S})$ therefore implies the sum of the remaining two eigenvalues vanishes: $\lambda_1(\bm{x}) + \lambda_2(\bm{x}) = 0$. Upon setting $\lambda(\bm{x}) =\lambda_1(\bm{x})$, these observations lead to the following expression for $\bm{S}(\bm{x})$:
\begin{align}
\bm{S}(\bm{x}) = \lambda(\bm{x})\,(\bm{e}_1(\bm{x})\bm{e}_1(\bm{x}) - \bm{e}_2(\bm{x})\bm{e}_{2}(\bm{x})).\label{delta_formula}
\end{align}
Note that the non-Killing assumption played no role in arriving at this formula.  

Now invoking the fact that $\bm{u}$ is  non-Killing reveals that the eigenvalue $\lambda(\bm{x})$ must be non-zero to ensure $\bm{S}(\bm{x}) \neq 0$. It follows that $\bm{S}(\bm{x})$ has three distinct eigenvalues $\{0,\lambda(\bm{x}),-\lambda(\bm{x})\}$, with $\lambda$ nowhere-vanishing, and that the corresponding unit eigenvectors are uniquely determined up to sign. In this way, the strain-rate tensor of any    non-Killing admissible $\bm{u}$ determines a unique, locally and smoothly-defined frame field $(\bm{e}_1,\bm{e}_2,\bm{e}_3)$ for choice of signs of the three unit vectors. The $8$-fold degeneracy in the frame definition may be reduced to $4$ by requiring right-handed orientation. Our convention will be that the unit vector $\bm{e}_1$ is an eigenvector of $\bm{S}$ with positive eigenvalue, while $\bm{e}_3$ is a null eigenvector of $\bm{S}$.

By Eqs.\,\eqref{kernel_condition} and \eqref{delta_formula},
\begin{align*}
\bm{S}\cdot \bm{B} = \lambda\,(\bm{e}_1\bm{e}_1 - \bm{e}_2\bm{e}_2)\cdot \bm{B} = \lambda\,B_1\,\bm{e}_1 - \lambda\,B_2\,\bm{e}_2 = 0,\quad \bm{B} = B_1\,\bm{e}_1 + B_2\,\bm{e}_2 + B_3\,\bm{e}_3.
\end{align*}
This requires $B_1 = B_2 = 0$ because $\lambda$ is nowhere-vanishing. Therefore the magnetic field $\bm{B}$ must have the following remarkably simple expression in the frame determined by $\bm{S}$:
\begin{align}
\bm{B} = B_3\,\bm{e}_3.\label{reduced_B}
\end{align}
In particular, $B_3 = \pm |\bm{B}|$. From this, the symmetry condition \eqref{modB_sym} for $|\bm{B}|$, and Eq.\,\eqref{QS_alt}, we arrive at the following condition on $\bm{u}$:
\begin{align*}
\bm{e}_3\cdot\nabla\bm{u} - \bm{u}\cdot \nabla\bm{e}_3 & = \frac{1}{B_3}\bm{B}\cdot \nabla\bm{u} - \bm{u}\cdot \nabla\bigg(\frac{\bm{B}}{B_3}\bigg)\\
& = \frac{1}{B_3}\bigg(\bm{B}\cdot \nabla\bm{u} - \bm{u}\cdot \nabla\bm{B}\bigg) - \bigg(\bm{u}\cdot \nabla B_3^{-1}\bigg)\bm{B}\\
& = 0.
\end{align*}

To complete the proof of Theorem \ref{necessary_conditions}, first observe that $\bm{\pi}_\parallel$ is given by
\begin{align*}
\bm{\pi}_\parallel = 1 - \frac{2\bm{S}\cdot\bm{S}}{\bm{S}:\bm{S}} = 1 - \frac{2\,\lambda^2 (\bm{e}_1\bm{e}_1 +\bm{e}_2\bm{e}_2)}{\lambda^2\cdot 2} = \bm{e}_3\bm{e}_3.
\end{align*}
It follows that
\begin{align*}
\bm{u}\cdot \nabla\bm{\pi}_{\parallel} - \nabla\bm{u}^T\cdot \bm{\pi}_{\parallel} - \bm{\pi}_{\parallel}\cdot \nabla\bm{u} & = \bm{u}\cdot\nabla (\bm{e}_3\bm{e}_3) - \nabla\bm{u}^T\cdot(\bm{e}_3\bm{e}_3) - (\bm{e}_3\bm{e}_3)\cdot\nabla\bm{u}\\
& = (\bm{u}\cdot\nabla\bm{e}_3)\bm{e}_3 + \bm{e}_3(\bm{u}\cdot\nabla\bm{e}_3) - (\bm{e}_3\cdot\nabla\bm{u})\bm{e}_3 - \bm{e}_3(\bm{e}_3\cdot\nabla\bm{u})\\
& = 0,
\end{align*}
which recovers Eq.\,\eqref{adm_nec_3}, as claimed.

%

The preceding argument successfully obtained the PDE constraints \eqref{adm_nec_1}-\eqref{adm_nec_3}. It also identified a notable representation for quasisymmetric magnetic fields $\bm{B}$ that dramatically simplifies the problem of finding a quasisymmetric $\bm{B}$ given an admissible $\bm{u}$. This is summarized in the following Lemma.
\begin{lemma}\label{B_lemma}
Fix a non-Killing $\bm{u}$. Assume there is a globally-defined null eigenvector, $\bm{e}_3$, for the $\bm{u}$-strain-rate tensor. Also assume $\bm{u}$ satisfies the PDE constraints \eqref{adm_nec_1}-\eqref{adm_nec_3}. The nowhere-vanishing vector field $\bm{B}$ is quasisymmetric with infinitesimal generator $\bm{u}$ if and only if $\bm{B} = B_3\,\bm{e}_3$, $B_3$ is nowhere-vanishing, and $B_3$ satisfies
\begin{align}
\bm{u}\cdot \nabla B_3 = 0,\quad \bm{e}_3\cdot\nabla B_3 = - B_3\,\nabla\cdot \bm{e}_3.\label{B3_eqn_lemma}
\end{align}
\end{lemma}
\begin{proof}
We have already shown that if $\bm{B}$ is nowhere-vanishing and quasisymmetric with non-Killing infinitesimal generator $\bm{u}$ then (a) $\bm{B} = B_3\,\bm{e}_3$ and (b) $\bm{u}\cdot \nabla B_3 = 0$. The condition $\bm{e}_3\cdot\nabla B_3 = - B_3\,\nabla\cdot \bm{e}_3$ follows simply from $ 0 = \nabla\cdot \bm{B} = \nabla\cdot (B_3\,\bm{e}_3)$.

It is still necessary to show that $\bm{B} = B_3\,\bm{e}_3$ is quasisymmetric when $B_3$ satisfies Eq.\,\eqref{B3_eqn_lemma}. But this follows from the direct calculations,
\begin{align}
\nabla\cdot\bm{B} & = \bm{e}_3\cdot\nabla B_3 + B_3\,\nabla\cdot\bm{e}_3 \label{comp_check_1}\\
\nabla\times(\bm{u}\times\bm{B}) &=\bm{B}\cdot \nabla\bm{u} - \bm{u}\cdot\nabla\bm{B} = B_3\,[\bm{e}_3,\bm{u}] -( \bm{u}\cdot \nabla B_3)\bm{e}_3 \label{comp_check_2}\\
(\nabla\times\bm{B})\times\bm{u} + \nabla(\bm{u}\cdot\bm{B}) &= \bm{u}\cdot\nabla\bm{B} + \nabla\bm{u}\cdot\bm{B} = (\bm{u}\cdot\nabla B_3)\,\bm{e}_3 + 2B_3\,\bm{S}\cdot \bm{e}_3 .\label{comp_check_3}
\end{align}
The right-hand-side of Eq.\,\eqref{comp_check_1} vanishes because $B_3$ satisfies Eq.\,\eqref{B3_eqn_lemma}. The right-hand-side of Eq.\,\eqref{comp_check_3} vanishes because $\bm{e}_3$ is a null eigenvector for $\bm{S}$ and $B_3$ satisfies Eq.\,\eqref{B3_eqn_lemma}. The right-hand-side of Eq.\,\eqref{comp_check_2} vanishes because
\begin{align*}
\bm{u}\cdot \nabla\bm{\pi}_{\parallel} - \nabla\bm{u}^T\cdot \bm{\pi}_{\parallel} - \bm{\pi}_{\parallel}\cdot \nabla\bm{u} & = \bm{u}\cdot\nabla (\bm{e}_3\bm{e}_3) - \nabla\bm{u}^T\cdot(\bm{e}_3\bm{e}_3) - (\bm{e}_3\bm{e}_3)\cdot\nabla\bm{u}\\
& = (\bm{u}\cdot\nabla\bm{e}_3)\bm{e}_3 + \bm{e}_3(\bm{u}\cdot\nabla\bm{e}_3) - (\bm{e}_3\cdot\nabla\bm{u})\bm{e}_3 - \bm{e}_3(\bm{e}_3\cdot\nabla\bm{u})\\
& = [\bm{u},\bm{e}_3]\,\bm{e}_3 + \bm{e}_3\,[\bm{u},\bm{e}_3]\\
& = 0,
\end{align*}
and $B_3$ satisfies Eq.\,\eqref{B3_eqn_lemma}.
\end{proof}

\section{Local admissibility\label{local_adm_sec}}
Next we provide a proof of Theorem \ref{local_admissibility_thm}.
\begin{proof}
The proof proceeds by constructing a smooth nowhere-vanishing solution $B_3$ of the equations $\bm{u}\cdot\nabla B_3 = 0$ and $\bm{e}_3\cdot\nabla B_3 = - B_3\,\nabla\cdot \bm{e}_3$ in a neighborhood of each point $\bm{x}$ where $\bm{u}(\bm{x})\times\bm{e}_3(\bm{x})\neq 0$. Given such a $B_3$, the vector field $\bm{B} = B_3\,\bm{e}_3$ is quasisymmetric with infinitesimal generator $\bm{u}$ by Lemma \ref{B_lemma}.

Suppose $\bm{x}\in Q$ and $\bm{u}(\bm{x})\times\bm{e}_3(\bm{x})\neq0 $. Then there is an open neighborhood $U$ of $\bm{x}$ on which $\bm{u}\times\bm{e}_3$ is nowhere-vanishing. By $[\bm{u},\bm{e}_3] = 0$ and the Frobenius theorem, there is a smooth function $\psi:U\rightarrow \mathbb{R}$ with nowhere-vanishing gradient such that $\bm{e}_3\cdot\nabla\psi = \bm{u}\cdot\nabla\psi =0$. Therefore there is a smooth, nowhere-vanishing function $\lambda$ such that $\bm{u}\times\bm{e}_3 = \frac{1}{\lambda}\nabla\psi$. If $\Omega$ denotes the Riemannian volume form on $Q$, the latter condition is equivalent to $\iota_{\bm{e}_3}\iota_{\bm{u}}\Omega = \frac{1}{\lambda}d\psi$. Lie-differentiating both sides of this formula with respect to $\bm{u}$ and using $\nabla\cdot\bm{u} = 0$, $[\bm{e}_3,\bm{u}]=0$, then implies
$\mathcal{L}_{\bm{u}}(1/\lambda)\,d\psi = 0$. Since $d\psi$ is nowhere-vanishing, the function $\lambda$ must therefore be constant along $\bm{u}$-lines, $\bm{u}\cdot\nabla\lambda = 0$. Now taking the curl of $\bm{u}\times (\lambda\bm{e}_3) = \nabla\psi$ gives
\begin{align*}
0 = \nabla\cdot (\lambda\bm{e}_3)\,\bm{u} + [\lambda\bm{e}_3,\bm{u}] =  \nabla\cdot (\lambda\bm{e}_3)\,\bm{u}.
\end{align*}
Since $\bm{u}$ is nowhere-vanishing on $U$ this implies $\nabla\cdot (\lambda\bm{e}_3) = 0$, or $\bm{e}_3\cdot\nabla\lambda = -\lambda\,\nabla\cdot\bm{e}_3$. We conclude that $B_3 = \lambda$ is the desired solution of $\bm{u}\cdot\nabla B_3 = 0$, $\bm{e}_3\cdot\nabla B_3 = - B_3\nabla\cdot\bm{e}_3$.

\end{proof}

\section{Global admissibility\label{global_adm_sec}}
Finally, we prove Theorem \ref{global_admissibility_thm}. 
\begin{proof}
Suppose $\bm{B}$ is quasisymmetric with non-Killing infinitesimal generator $\bm{u}$. Also assume that $\bm{u}\times\bm{B}$ is nowhere vanishing and that both $\bm{u}$ and $\bm{B}$ are tangent to $\partial Q$. Theorem \ref{necessary_conditions} immediately implies $\bm{u}$ satisfies \eqref{adm_nec_1}-\eqref{adm_nec_3}. The equation \eqref{kernel_condition} immediately implies that $\bm{e}_3 = \bm{b}$ is a nowhere-vanishing global section of the line bundle $\text{ker}\,\bm{S}$. The orthogonal projector onto $\text{ker}\,\bm{S}$ is therefore $\bm{\pi}_\parallel = \bm{b}\bm{b}$, which shows $(1-\bm{\pi}_\parallel)\cdot\bm{u} = \bm{b}\times (\bm{u}\times\bm{b})$ is nowhere-vanishing in $Q$. To establish existence of $\psi$, consider the vector field $\bm{w} = \bm{u}\times\bm{B}$. The curl of $\bm{w}$ vanishes. Moreover, the line integral of $\bm{w}$ around any closed loop in $\partial Q$ vanishes because $\bm{u}$ and $\bm{B}$ provide a basis for vectors tangent to $\partial Q$. Since $Q$ is homotopic to either component of $\partial Q$, it follows that $\bm{w} = \nabla \psi$ for some single-valued function $\psi:Q\rightarrow \mathbb{R}$. Clearly $\bm{u}\cdot \nabla\psi = \bm{B}\cdot \nabla\psi = 0$ and $\nabla\psi$ is nowhere-vanishing.


Now for the converse. Let $\bm{E} = \bm{\pi}_\parallel\cdot (\bm{u}\times\nabla\psi)$. Clearly $\bm{E}$  takes values in $\text{ker}\,\bm{S}$. We claim $\bm{E}$ is also nowhere-vanishing. For suppose it vanished at $\bm{x}\in Q$. Consider an open neighborhood $U$ of $\bm{x}$ and an orthonormal frame $(\bm{e}_1,\bm{e}_2,\bm{e}_3)$ defined on $U$ such that $\bm{\pi}_\parallel = \bm{e}_3\bm{e}_3$. Vanishing of $\bm{E}$ at $\bm{x}$  implies $\bm{e}_3(\bm{x})\cdot (\bm{u}(\bm{x})\times\nabla\psi(\bm{x})) = (\bm{e_3}(\bm{x})\times\bm{u}(\bm{x}))\cdot \nabla\psi(\bm{x}) = 0$. The vector $\bm{e_3}(\bm{x})\times\bm{u}(\bm{x})$ is not zero because $(1-\bm{\pi}_\parallel)\cdot\bm{u}$ is nowhere-vanishing. Moreover, $\bm{\pi}_\parallel\cdot\nabla\psi = 0$ implies $\bm{e}_3\cdot\nabla\psi = 0$ in $U$. Thus $\nabla\psi(\bm{x})$ is orthogonal to the three linearly-independent vectors $\bm{u}(\bm{x})$, $\bm{e}_3(\bm{x})$, and $\bm{e}_3(\bm{x})\times\bm{u}(\bm{x})$, contradicting the fact that $\nabla\psi$ is nowhere-vanishing. We conclude that $\bm{E}$ is nowhere-vanishing, as claimed, and $\bm{e}_3 = \bm{E}/|\bm{E}|$ is a globally-defined unit vector field that takes values in $\text{ker}\,\bm{S}$. In addition, $\bm{e}_3\times\bm{u}$ is nowhere-vanishing and $\bm{e}_3$ is tangent to $\partial Q$. Repeating the argument from the proof of Theorem \ref{local_admissibility_thm} now shows that $\lambda = |\nabla\psi|/|\bm{u}\times\bm{e}_3|$ is a nowhere-vanishing, globally-defined solution of $\bm{u}\cdot\nabla \lambda = 0$, $\bm{e}_3\cdot\nabla \lambda = - \lambda\,(\nabla\cdot\bm{e}_3)$. Lemma \ref{B_lemma} therefore implies that $\bm{B} = \lambda\,\bm{e}_3$ is quasisymmetric with infinitesimal generator $\bm{u}$. Moreover, $\bm{u}\times\bm{B}  =\lambda \bm{u}\times\bm{e}_3$ is nowhere-vanishing in $Q$ and $\bm{B} = \lambda\,\bm{e}_3$ is tangent to $\partial Q$.

\end{proof}

\section{Discussion\label{discussion_sec}}
The proofs of Theorem \ref{local_admissibility_thm} and Theorem \ref{global_admissibility_thm} involve explicit construction of quasisymmetric $\bm{B}$ with a given admissible non-Killing infinitesimal generator $\bm{u}$. These constructions solve ``half" of the problem (step (2) from the introduction) of constructing genuinely-three-dimensional quasisymmetric $\bm{B}$. The other ``half" problem remains open --- the existence question for smooth non-Killing admissible $\bm{u}$ needs to be addressed. A subsequent publication will apply Cartan's method of prolongation to involution, along with the Cartan-K\"{a}hler theorem (see \cite{Bryant_EDS_1991}), to resolve at least the local existence problem. 

The admissibility problem can also be studied for weak quasisymmetry, as defined in \cite{Rodriguez_2020}. Details of such analysis will appear in future publications.


Let $\bm{e}_3$ be a globally-defined unit vector field $\bm{e}_3$ in the null space of $\bm{S}$. The proof of Theorem \ref{necessary_conditions} in Section \ref{adm_deriv} shows that the condition \eqref{adm_nec_3} is equivalent to $[\bm{e}_3,\bm{u}] = 0$. The latter condition may be expressed in terms of the coefficients of $\bm{S}$ in Cartesian coordinates $(x,y,z)$ as follows. Non-killing $\bm{u}$ implies some principal cofactor for $\bm{S}$ is non-zero. Thus, in a region where, say, the $zz$-cofactor for $\bm{S}$ is non-zero, an explicit formula for $\bm{e}_3$ is
\begin{align*}
\bm{v} = \begin{pmatrix} \bm{S}_{xy}\bm{S}_{yz} - \bm{S}_{yy}\bm{S}_{xz}\\
\bm{S}_{xz}\bm{S}_{yx} - \bm{S}_{xx}\bm{S}_{yz}\\ \bm{S}_{xx}\bm{S}_{yy} - \bm{S}_{xy}^2\end{pmatrix},\quad \bm{e}_3 = \bm{v}/|\bm{v}|.
\end{align*}
This allows re-writing the condition $[\bm{u},\bm{e}_3] = 0$ as an ordinary PDE in the components of $\bm{u}$, namely $\bm{v}\times[\bm{u},\bm{v}] = 0$.

By repeating the proof of Theorem \ref{global_admissibility_thm} with $\psi$ replaced with $\psi^\prime = f(\psi)$, where $f$ is any smooth monotone single-variable function, we find that the most general $\bm{B}$ with a given non-Killing $\bm{u}$ as its infinitesimal generator is given by
\begin{align*}
\bm{B} = \frac{|f^\prime(\psi)||\nabla\psi|}{|\bm{u}\times\bm{e}_3|}\bm{e}_3.
\end{align*}
Thus, once the symmetry is known, $\bm{B}$ is determined by a single free flux function $f(\psi)$. Apparently the space of genuinely-three-dimensional quasisymmetric $\bm{B}$ is not much larger than the space of non-Killing admissible $\bm{u}$! The representation \eqref{B_stream_function_rep} should be compared with the well-known representation given in Lemma \ref{killing_case}.


An alternative route to the equations found here for a non-Killing $\bm{u}$ is via Theorems X.2 and X.3 of \cite{BKM_2020}.  $\bm{B}$ can be eliminated from these by using $\bm{B}=(\bm{u}\times \nabla \psi + C\bm{u})/|\bm{u}|^2$ with $C = \bm{u}\cdot \bm{B}$, and we had proved $\bm{u}\cdot\nabla C=0$.  Then $\psi$ can be eliminated by $\nabla \psi = |\bm{B}| \bm{u}\times \bm{b}$.

The global characterization of admissible $\bm{u}$ provided by Theorem \ref{global_admissibility_thm} requires linear independence of $\bm{u}$ and $\bm{e}_3$. Since any quasisymmetric $\bm{B}$ is parallel to $\bm{e}_3$ and it is well-known that $\bm{u}$ is parallel to $\bm{B}$ along a magnetic axis, Theorem \ref{global_admissibility_thm} cannot be applied naively when the spatial domain $Q$ is a solid torus instead of a toroidal annulus. Obtaining a global characterization of three-dimensional admissible $\bm{u}$ in the presence of magnetic axes will be the subject of future work.

The arguments presented in this Article assumed the spatial domain $Q$ is a region in $\mathbb{R}^3$ and that the metric tensor on $Q$ is flat. They generalize readily to the more general situation where $Q$ is a Riemannian $3$-manifold with possibly-non-flat metric tensor $g$. The strain-rate tensor then becomes $\bm{S} = g^{-1}[\mathcal{L}_{\bm{u}}g]g^{-1}/2$, where $\mathcal{L}_{\bm{u}}g$ denotes the Lie derivative of $g$ along $\bm{u}$ and $g^{-1}$ indicates raising indices using the metric $g$. Such non-flat metrics have become interesting to consider, even in terrestrial settings, due to recent results from \cite{Cardona_Duignan_Perrella_2023} showing that smooth non-symmetric solutions of magnetohydrostatics with non-trivial rotational transform exist for non-flat metrics that approximate the flat metric with arbitrary precision.

\section*{Acknowledgements}

This work was supported by a grant from the Simons Foundation (601970, RSM). For the purpose of open access, the authors have applied a Creative Commons Attribution (CC-BY) licence to any Author Accepted Manuscript version arising from this submission. This work was also supported by U.S. Department of Energy Contract No. DE-FG05-80ET-53088 (JWB).

\bibliographystyle{jpp}



\providecommand{\noopsort}[1]{}\providecommand{\singleletter}[1]{#1}%

\makeatletter
\def\fps@table{h}
\def\fps@figure{h}
\makeatother

\section*{Appendix}

\subsection{An illuminating counterexample}
Consider the toroidal annulus $Q = [-1,1]\times S^1\times S^1\ni (r,\theta,\phi)$. Equip $Q$ with the metric tensor $g = dr^2 + d\theta^2 + d\phi^2$. Define the vector fields $\bm{e}_3,\bm{u}$ on $Q$ according to
\begin{align*}
\bm{u} = \partial_\phi,\quad \bm{e}_3 = \frac{(r+1)(r-1)}{\sqrt{(r+1)^2(r-1)^2 + 1}}\partial_r + \frac{1}{\sqrt{(r+1)^2(r-1)^2 + 1}}\,\partial_\theta.
\end{align*}
It's easy to see that $\bm{e}_3$ is a unit vector, $[\bm{e}_3,\bm{u}]=0$, $\nabla\cdot\bm{u} = 0$, and $\bm{u},\bm{e}_3$ are everywhere linearly-independent. The $2$-planes spanned by $\bm{e}_3,\bm{u}$ are given as the vanishing locus of the $1$-form $\lambda$ on $Q$ defined by
\begin{align}
\lambda = -dr + (r+1)(r-1)d\theta.\label{example_lambda}
\end{align}
The form $\lambda$ satisfies the Frobenius integrability condition $\lambda\wedge d\lambda = 0$. It is nowhere-vanishing on $Q$.

One way to state the Frobenius integrability theorem is $\lambda\wedge d\lambda = 0$ implies there is a locally-defined smooth positive function $h$ such that $h\lambda$ is closed. The function $h$ is known as an \textbf{integrating factor}. If there was a smooth function $\psi$ with nowhere-vanishing gradient such that $\bm{u}\cdot\nabla\psi = \bm{e}_3\cdot\nabla\psi = 0$ then $d\psi$ would be parallel to $\lambda$, thus implying existence of a global integrating factor $h$. The following argument shows that $h$ cannot be globally defined in this example, and therefore neither can $\psi$.

Suppose $h$ were globally defined. Then there would be constants $a_\theta,a_{\phi}\in\mathbb{R}$ and a smooth function $\chi:Q\rightarrow\mathbb{R}$ such that
\begin{align*}
h\lambda = a_\theta\,d\theta + a_\phi\,d\phi + d\chi.
\end{align*}
The integrals of $h\lambda$ along the parameterized curves
\begin{align*}
c_\theta(\zeta) = (1,\zeta,0),\quad c_\phi(\zeta) = (1,0,\zeta),
\end{align*}
each vanish because $\lambda$ vanishes when pulled back to either boundary component of $Q$. Thus,
\begin{align*}
0=\int_{c_\theta}h\lambda=2\pi\,a_\theta,\quad 0 = \int_{c_\phi}h\lambda = 2\pi\,a_\phi,
\end{align*}
which shows $h\lambda = d\chi$. In particular, the differential $d\chi$ is nowhere-vanishing on $Q$, $\chi$ is constant on each of the boundary components of $Q$, and the level sets of $\chi$ are the integral manifolds for $\lambda$. By nowhere-vanishing $d\chi$, the level sets of $\chi$ are all embedded submanifolds in $Q$ with a common diffeomorphism type. But since each boundary component of $Q$ is simultaneously a $2$-torus and a level set for $\chi$ it follows that each level set of $\chi$ is an embedded $2$-torus.

We will now show that, except for the boundary components of $\partial Q$, each integral manifold for $\lambda$ is also diffeomorphic to a cylinder $\mathbb{R}\times S^1$. Since cylinders are not diffeomorphic to 2-tori, this contradiction implies that there can be no global $h$.

The integral manifolds for $\lambda$ can be computed explicitly from the definition of $\lambda$,  Eq.\,\eqref{example_lambda}. Consider the system of ordinary differential equations in the annulus $[-1,1]\times S^1\ni (r,\theta)$ defined by
\begin{align}
\dot{r} = (r+1)(r-1),\quad \dot{\theta}=1.\label{basic_ode}
\end{align}
It is easy to see that the integral manifolds for $\lambda$ are each of the form $\gamma\times S^1$, where $\gamma$ is an integral curve for \eqref{basic_ode} and $S^1$ denotes the $\phi$-axis. In other words, the integral manifolds for $\lambda$ are surfaces of revolution generated by integral curves of \eqref{basic_ode}. The general solution for \eqref{basic_ode} is
\begin{align*}
r(t) = \frac{1-\left[\frac{1-r_0}{1+r_0}\right]\,e^{2t}}{1+\left[\frac{1-r_0}{1+r_0}\right]\,e^{2t}},\quad \theta(t) = \theta_0 + t\text{ mod }2\pi,\quad (r_0,\theta_0)\in [-1,1]\times S^1.
\end{align*}
For $r_0\in (-1,1)$ the derivative $\dot{r}(t)$ is strictly negative for all $t\in \mathbb{R}$ and $\lim_{t\rightarrow\infty}r(t) = -1$, $\lim_{t\rightarrow-\infty}r(t) =1$. Therefore $t\mapsto r(t)$ defines a diffeomorphism $\mathbb{R}\rightarrow (-1,1)$. It follows that the image of $\gamma(t) =  (r(t),\theta(t))$ is an immersed $1$-manifold diffeomorphic to $\mathbb{R}$. Hence, the revolution $\gamma\times S^1$ is diffeomorphic to $\mathbb{R}\times S^1$, as claimed.

\subsection{Replacing (II.d) with (II.d')}
Here we prove that Theorem \ref{global_admissibility_thm} is still valid if property (b) is replaced with property (b'). In other words, the only global admissibility condition that must be appended to the local conditions \eqref{adm_nec_1}-\eqref{adm_nec_3} is that the leaves of the foliation tangent to $\text{span}(\bm{e}_1,\bm{u})$ should be compact. Note the appearance of non-compact leaves in the counterexample above.

We argue by establishing the following more general result.
\begin{theorem}\label{thm_compact_leaves}
For $n\in\mathbb{N}$, let $M$ be a compact connected $(n+1)$-manifold with boundary that embeds in $\mathbb{R}^{n+1}$. Let $\mathcal{F}$ be a codimension-$1$ foliation of $M$ such that
\begin{itemize}
\item[1.] Each leaf in $\mathcal{F}$ is compact
\item[2.] $\mathcal{F}$ is transitive. (For each leaf $L$ and $x,y\in L$ there is a leaf-preserving diffeomorphism $F:M\rightarrow M$ with $y=F(x)$.)
\end{itemize}
There is a smooth function $\psi:M\rightarrow\mathbb{R}$ with $d\psi$ nowhere-vanishing such that the leaves of $\mathcal{F}$ are precisely the level sets of $\psi$.
\end{theorem}
\begin{remark}
The Theorem only applies to foliations with leaves that each have empty boundary. These are the foliations relevant to Theorem \ref{global_admissibility_thm} by properties (I.b) and (I.c).
\end{remark}
\begin{remark}
When $M$ is a toroidal annulus in $\mathbb{R}^{3}$ and the foliation $\mathcal{F}$ is generated by a pair of commuting vector fields that are everywhere linearly-independent it is easy to see that property 2 is satisfied. In fact, one can prove (non-constructively) that every foliation on any manifold with boundary is transitive. 
\end{remark}

Theorem \ref{thm_compact_leaves} applies to the present work as follows. 
\begin{theorem}[global admissibility (alternate)]\label{global_admissibility_thm_alt}
Assume $Q$ is diffeomorphic to the toroidal annulus $S^1\times S^1\times [0,1]$, where $S^1 = \mathbb{R}/(2\pi\mathbb{Z})$ denotes the circle. Properties (I) and (II) for a vector field $\bm{u}$ are equivalent.
\begin{itemize}
\item[(I.a)] $\bm{u}$ is admissible and non-Killing with corresponding quasisymmetric magnetic field $\bm{B}:Q\rightarrow\mathbb{R}^3$.
\item[(I.b)] $\bm{u}\times \bm{B}$ is nowhere-vanishing in $Q$.
\item[(I.c)] Both $\bm{u}$ and $\bm{B}$ are tangent to $\partial Q$.
\end{itemize}

\begin{itemize}
\item[(II.a)] $\bm{u}$ is non-Killing and satisfies Eqs.\,\eqref{adm_nec_1}-\eqref{adm_nec_3}.
\item[(II.b)] $(1-\bm{\pi}_\parallel)\cdot \bm{u}$ is nowhere-vanishing in $Q$.
\item[(II.c)] Both $\bm{u}$ and $\text{\normalfont im}\,\bm{\pi}_\parallel$ are tangent to $\partial Q$.
\item[(II.d')] The leaves of the foliation integrating $\text{\normalfont span}(\bm{u})\oplus \text{ker}\,\bm{S}$ are each compact.
\end{itemize}
%
%
%
\end{theorem}
\begin{proof}
For our spatial domain $Q$ we have considered regions in $\mathbb{R}^{(n+1)}$, $n=2$,  diffeomorphic to the toroidal annulus $S^1\times S^1\times [0,1]$. The space $M\equiv Q$ is therefore a compact connected $(n+1)$-manifold that embeds in $\mathbb{R}^{n+1}$.

If $\bm{B}$ is quasisymmetric with non-Killing infinitesmial generator $\bm{u}$, $\bm{u}\times\bm{B}$ is nowhere-vanishing, and both $\bm{u}$ and $\bm{B}$ are tangent to $\partial Q = \partial M$, then Theorem \ref{global_admissibility_thm} implies property (II) from the statement of Theorem \ref{global_admissibility_thm}. In particular, there is a function $\psi$ with $\nabla\psi$ nowhere-vanishing such that $\bm{u}\cdot\nabla\psi =0$ and $\bm{e}_3 \cdot\nabla \psi = 0$, where $\bm{e}_3 = \bm{b}$. The level sets of $\psi$ are compact connected embedded submanifolds diffeomorphic to the $2$-torus $S^1\times S^1$. If $\Lambda$ is one such level set then $\bm{u}\cdot\nabla\psi = \bm{e}_3\cdot\nabla\psi = 0$ implies $\Lambda$ is tangent to the foliation $\mathcal{F} =\text{span}(\bm{u},\bm{e}_3)$. Thus, $\Lambda$ is an integral submanifold for  $\mathcal{F}$.
Let $L$ be any connected integral submanifold that contains $\Lambda$ and choose $p\in L$, $p^\prime\in\Lambda$.
Since $L$ is connected there must be a smooth curve $\gamma:[0,1]\rightarrow L$ with $\gamma(0) = p,\gamma(1) = p^\prime$. The velocity of $\gamma$, $\gamma^\prime(t)$, is contained in $\text{span}(\bm{u}(\gamma(t)),\bm{e}_3(\gamma(t)))$. The derivative of $\psi$ along $\gamma$ is therefore
\begin{align*}
\frac{d}{dt}\psi(\gamma(t)) = \gamma^\prime(t)\cdot\nabla\psi(\gamma(t)) = 0,
\end{align*}
since $\bm{u}\cdot\nabla\psi = \bm{e}_3\cdot\nabla\psi = 0$. It follows that $\psi(p) = \psi(p^\prime)$. In other words, $p\in\Lambda$. Since $p\in L$ is arbitrary, this implies $L = \Lambda$. It follows that every leaf of the foliation is a level set of $\psi$. Since the level sets of $\psi$ are compact, we conclude that $\bm{u}$ satisfies property $(II.d')$.

Now suppose that $\bm{u}$ satisfies properties (II.a), (II.b), (II.c), and (II.d'). As mentioned in the remark above, the foliation $\mathcal{F} $ that integrates $\text{span}(\bm{u})\oplus\text{ker}\,\bm{S}$ is transitive. By property (II.d') and Theorem \ref{thm_compact_leaves}, there is a smooth function $\psi:M\rightarrow\mathbb{R}$ with $\nabla\psi$ nowhere-vanishing such that $\bm{u}\cdot\nabla\psi = \bm{e}_3\cdot\nabla\psi = 0$, where $\bm{e}_3$ is a globally-defined unit vector field that takes values in $\text{ker}\,\bm{S}$. (The proof of Theorem \ref{global_admissibility_thm} explains why we can assume such an $\bm{e}_3$ exists.) The latter is just property (II.d) from Theorem \ref{global_admissibility_thm}. Therefore $\bm{u}$ is admissible and any quasisymmetric $\bm{B}$ with $\bm{u}$ as its infinitesimal generator has $\bm{u}\times\bm{B}$ nowhere-vanishing and $\bm{u},\bm{B}$ tangent to $\partial M = \partial Q$.
\end{proof}

\begin{remark}
Let $\bm{e}_3$ be a globally-defined unit vector field with $\bm{S}\cdot\bm{e}_3 = 0$. One case in which property (II.d') is clearly satisfied occurs when the integral curves of the vector fields $\bm{u}$ and $\bm{e}_3$ are each periodic. The integral curves of $\bm{u}$ satisfying property (II), and therefore property (I), will be always be periodic by arguments from \cite{BKM_2020}. Since $\bm{b} = \pm \bm{e}_3$, periodic integral curves for $\bm{e}_3$ corresponds to quasisymmetric $\bm{B}$ with all $\bm{B}$-lines closed. 
\end{remark}

Now for a proof of Theorem \ref{thm_compact_leaves}.

First we make a quotient space out of $M$ from $\mathcal{F}$. We say that points $x,y\in M$ are equivalent if they belong to the same leaf. The quotient of $M$ by this equivalence relation is the space of leaves of $\mathcal{F}$. It will be convenient to identify $\mathcal{F}$ with its leaf space. There is therefore a well-defined quotient map $\pi:M\rightarrow\mathcal{F}$ sending each point $x\in M$ to the leaf that contains it. We topologize $\mathcal{F}$  by equipping it with the quotient topology, wherein $U\subset \mathcal{F}$ is open in $\mathcal{F}$ if $\pi^{-1}(U)$ is open in $M$. Note that if $L$ is any leaf then $\pi^{-1}(\{L\}) = L$. When $\mathcal{F}$ is equipped with this topology, $\pi:M\rightarrow\mathcal{F}$ is continuous. Since $M$ is compact and connected and $\pi$ is a continuous surjection, it immediately follows that $\mathcal{F}$ is compact and connected.

We aim to prove that $\mathcal{F}$ is a second-countable Hasudorff topological space that can be given the structure of a smooth connected compact $1$-manifold with non-empty boundary and that $\pi$ is a smooth submersion with respect to this structure. This will prove Theorem \ref{thm_compact_leaves} because every compact connected $1$-manifold with non-empty boundary is diffeomorphic to $[0,1]$. If $\Psi:\mathcal{F}\rightarrow[0,1]$ is such a diffeomorphism then function $\psi$ guaranteed by Theorem \ref{thm_compact_leaves} is
\begin{align*}
\psi = \Psi\circ\pi : M\rightarrow [0,1].
\end{align*}

The proof will proceed by constructing an appropriate atlas of charts on $\mathcal{F}$ out of $\mathcal{F}$-adapted charts on $M$.
\begin{definition}\label{adapted_def}
Given $x\in M$, a coordinate chart $(U,\varphi)$ is $\mathcal{F}$-\textbf{adapted} at $x$ if there is $\epsilon > 0 $ such that
\begin{itemize}
\item[1.] $\varphi(x) = 0$
\item[2.] $\varphi(U) = (-\epsilon,\epsilon)^n\times I$, where \begin{align*}
I= \begin{cases}
[0,\epsilon), & \text{if }x\in\partial M\\
(-\epsilon,\epsilon) & \text{otherwise}
\end{cases}
\end{align*}
\item[3.] For each leaf $L$ that meets $U$ there is a unique subset $\mathcal{C}_L\subset I$ with $\varphi(U\cap L) = (-\epsilon,\epsilon)^n\times\mathcal{C}_L$.
\end{itemize}
\end{definition}
\noindent As we will show, each leaf $L$ enters a given adapted chart $(U,\varphi)$ at most once. We formalize this result by  proving that $\mathcal{C}_L$ is a singleton. There is therefore a well-defined real-valued function $h$ on the set of leaves $L$ that meet $U$ defined according to $\mathcal{C}_L = \{h(L)\}$. The atlas on $\mathcal{F}$ will comprise charts of the form $(\pi(U),h)$. We will first show that these charts are well-defined before demonstrating that the corresponding smooth structure on $\mathcal{F}$ has the desired properties.

In order for $(\pi(U),h)$ to be a chart on $\mathcal{F}$ it must certainly be true that $\pi(U)$ is open in $\mathcal{F}$. We will establish this first result by showing that $\pi$ is an open map. To do so, we first prove an algebraic Lemma.
\begin{lemma}\label{growing_along_leaves}
Let $\text{Diff}_{\mathcal{F}}(M)$ denote the set of diffeomorphisms $F:M\rightarrow M$ such that $F(L) = L$ for each leaf $L\in\mathcal{F}$. If $A\subset M$ is any subset then
\begin{align*}
\pi^{-1}(\pi(A)) = \bigcup\{F(A)\mid F\in\text{Diff}_{\mathcal{F}}(M)\}.
\end{align*}
\end{lemma}
\begin{proof}
Let $S_1 = \pi^{-1}(\pi(A))$ and $S_2 = \bigcup\{F(A)\mid F\in\text{Diff}_{\mathcal{F}}(M)\}$. We will show $S_1\subset S_2$ and $S_2\subset S_1$.

If $y\in S_2$ there is some $F\in\text{Diff}_{\mathcal{F}}(M)$ such that $y\in F(A)$. Choose $x\in A$ such that $F(x) = y$. Since $F$ preserves every leaf in $\mathcal{F}$, $\pi(y) = \pi(F(x)) = \pi(x)\in\pi(A)$. This shows $S_2\subset S_1$.

If $y\in S_1$ there is some $x\in A$ such that $\pi(y) = \pi(x)$. Because $\mathcal{F}$ is transitive there is $F\in\text{Diff}_{\mathcal{F}}(M)$ with $y = F(x)\in F(A)\subset S_2$. Thus, $S_1\subset S_2$.
\end{proof}

\begin{lemma}\label{open_map_lemma}
The map $\pi:M\rightarrow\mathcal{F}$ is open.
\end{lemma}
\begin{proof}
We must show that the image $\pi(U)$ of any open $U\subset M$ is open in $\mathcal{F}$. By definition of the quotient topology on $\mathcal{F}$, this is equivalent to showing that $\pi^{-1}(\pi(U))$ is open in $M$. But by Lemma \ref{growing_along_leaves}, $\pi^{-1}(\pi(U)) = \bigcup \{F(U)\mid F\in\text{Diff}_{\mathcal{F}}(M)\}$ is a union of the open sets $F(U)$.
\end{proof}

In order for $(\pi(U),h)$ to be a chart on $\mathcal{F}$ it must also be true that $h:\pi(U)\rightarrow \mathbb{R}$ is a well-defined homeomorphism onto its image. The following technical results provide a framework for rigorously establishing this property. 

Fix an $\mathcal{F}$-adapted chart $(U,\varphi)$ on $M$. Define $\mathcal{F}_U = \pi(U)$. For $L\in\mathcal{F}_U$ let $\mathcal{C}_L\subset\mathbb{R}$ denote the set such that $\varphi(U\cap L) = (-\epsilon,\epsilon)^n\times\mathcal{C}_L$, cf. Definition \ref{adapted_def}. The following argument builds to a proof that $\mathcal{C}_L$ is a singleton. First we show that $\mathcal{C}_L$ is discrete.
\begin{lemma}\label{discrete_lemma}
For each $L\in\mathcal{F}_U$ the set $\mathcal{C}_L$ is discrete.
\end{lemma}
\begin{proof}
Recall that we assume all leaves in $\mathcal{F}$ are compact. Also recall that each leaf of a foliation is an immersed submanifold. Consequently, $L$ is a compact immersed submanifold of $M$. By proposition 5.21 in \cite{Lee_2003}, it follows that $L$ is an embedded submanifold of $M$.

Since $U\subset M$ is open, $L\cap U$ is an embedded submanifold of $U$ with dimension $n$. Since $\varphi$ is a diffeomorphism, $\varphi(L\cap U) = (-\epsilon,\epsilon)^n\times\mathcal{C}_L$ is an embedded submanifold of $\varphi(U)$ with dimension $n$. For $c\in\mathcal{C}_L$, $(-\epsilon,\epsilon)^n\times\{c\}$ is also an embedded submanifold of $\varphi(U)$ with dimension $n$. Clearly, $(-\epsilon,\epsilon)^n\times\{c\}\subset (-\epsilon,\epsilon)^n\times\mathcal{C}_L$. So $(-\epsilon,\epsilon)^n\times\{c\}$ is an open subset of $(-\epsilon,\epsilon)^n\times\mathcal{C}_L$.

Since $(-\epsilon,\epsilon)^n\times\mathcal{C}_L$ is an embedded submanifold of $\varphi(U)$, its open sets are restrictions of open subsets of $\varphi(U)$. Since $(-\epsilon,\epsilon)^n\times\{c\}$ is open in $(-\epsilon,\epsilon)^n\times\mathcal{C}_L$ it follows that there is an open set $V\subset\varphi(U)$ such that 
\begin{align*}
(-\epsilon,\epsilon)^n\times\{c\} = \bigg((-\epsilon,\epsilon)^n\times\mathcal{C}_L\bigg)\cap V.
\end{align*}
Choose $\delta >0$ small enough so that $R_\delta(c) = (0,c) + (-\delta,\delta)^{n+1}\subset V$. We have
\begin{align*}
\bigg((-\epsilon,\epsilon)^n\times\mathcal{C}_L\bigg)\cap R_\delta(c)\subset  \bigg((-\epsilon,\epsilon)^n\times\mathcal{C}_L\bigg)\cap V = (-\epsilon,\epsilon)^n\times\{c\},
\end{align*}
and also
\begin{align*}
\bigg((-\epsilon,\epsilon)^n\times\mathcal{C}_L\bigg)\cap R_\delta(c) = (-\delta,\delta)^{n}\times \bigg(\mathcal{C}_L\cap (-\delta + c, c + \delta) \bigg).
\end{align*}
So $\mathcal{C}_L\cap (-\delta + c, c + \delta) = \{c\}$. In other words the $c\in\mathcal{C}_L$ are isolated by open intervals, as claimed.
\end{proof}

Next we show leaves that enter an adapted chart at least twice must satisfy a nesting property that will ultimately prove to be incompatible with discreteness of $\mathcal{C}_L$.
\begin{lemma}\label{nesting_lemma}
Suppose $L\in\mathcal{F}_U$ and there exists $a,b\in\mathcal{C}_L$ with $a <  b$. 
\begin{itemize}
\item[(a)] For each $a< m < b$ with $m\not\in \mathcal{C}_L$ there is a leaf $L^\prime$ and an $a<\tilde{m} < b$, distinct from $m$, such that $m,\tilde{m}\in \mathcal{C}_{L^\prime}$. 
\item[(b)] There is an $L^\prime\in\mathcal{F}_U$, distinct from $L$, and $a^\prime,b^\prime\in\mathcal{C}_{L^\prime}$ such that $a < a^\prime < b^\prime < b$ and $(b^\prime - a^\prime) < 3(b-a)/4$.
\end{itemize}
\end{lemma}
\begin{proof}
First we will show (a). 
The subset $\varphi^{-1}((-\epsilon,\epsilon)^n\times \{m\})\subset M$ is an integral manifold for the $n$-plane distribution associated with $\mathcal{F}$. Let $L^\prime$ denote the unique leaf in $\mathcal{F}$ (i.e. maximal integral manifold) that contains $\varphi^{-1}((-\epsilon,\epsilon)^n\times \{m\})$. By construction, $L^\prime$ is disjoint from $L$, $L^\prime\in\mathcal{F}_U$, and $m\in \mathcal{C}_{L^\prime}$. Suppose there were not a distinct $\tilde{m}\in (a,b)$ with $\tilde{m}\in\mathcal{C}_{L^\prime}$. Then the curve $\kappa: [a,b]\rightarrow M$ given by $\kappa(t) = \varphi^{-1}(0,t)$, $t\in[a,b]$, would have the following properties. 
\begin{itemize}
\item[(1)] $\kappa(t) \in L^\prime$ if and only if $t = m$. 
\item[(2)] $\kappa^\prime(m)$ is transverse to $L^\prime$.
\item[(3)] $\kappa(a), \kappa(b)\in L$.
\end{itemize} 
But this $\kappa$ cannot exist for the following reason. Embed $L$ and $L^\prime$ in $\mathbb{R}^{n+1}$. By the Jordan-Brower separation Theorem, the set $\mathbb{R}^{n+1} - L^\prime$ is the disjoint union of two connected open sets $C_+,C_-$. Moreover, there is a smooth function $f:\mathbb{R}^{n+1}\rightarrow\mathbb{R}$ such that $0$ is a regular value for $f$, $f^{-1}(\{0\}) = L^\prime$, $f\mid C_+ > 0$, and $f\mid C_- < 0$. Since $\kappa$ transversally intersects $L^\prime$ when $t=m$ the continuous function $F:[a,b]\rightarrow \mathbb{R}:t\mapsto f(\kappa(t))$ changes sign at $t = m$. Since $\kappa$'s intersection with $L^\prime$ is unique the restricted curves $\kappa\mid [a,m)$, $\kappa\mid (m,b]$ do not intersect $L^\prime$. This implies $F\mid [a,m)$ and $F\mid (m,b]$ have distinct constant signs. In particular, the signs of $F(a)$ and $F(b)$ differ. Therefore $\kappa(a),\kappa(b)\in L$ must lie in distinct connected components of $\mathbb{R}^{n+1} - L^\prime$. This is impossible because (I) connectedness of $L$ implies there is continuous curve $c:[0,1]\rightarrow L$ with $c(0) = \kappa(a)$, $c(1) = \kappa(b)$, (II) the value of $f\circ c$ must change sign, and (III) the intermediate value theorem implies there is some $\lambda\in[0,1]$ with $f(c(\lambda)) = 0$, which says $c(\lambda)\in L\cap L^\prime = \emptyset$. Thus, there must be $\tilde{m}\in (a,b)$, distinct from $m$, with $\tilde{m}\in \mathcal{C}_{L^\prime}$.

Now we show (b). Set $m_0 = (a + b)/2$. If $m_0\in\mathcal{C}_L$ use discreteness of $\mathcal{C}_L$ to find a nearby $m_0^\prime\not\in \mathcal{C}_L$ with $|m_0 - m_0^\prime| < (b-a)/4$ and set $m = m_0^\prime$. Otherwise set $m = m_0$. Note that $m\in(a,b) - \mathcal{C}_L$ and that $m$ is at least a distance $(b-a)/4$ from either endpoint $a,b$. Part (a) therefore implies there is an $a<\tilde{m}<b$, distinct from $m$, and a leaf $L^\prime$ such that $m,\tilde{m}\in\mathcal{C}_{L^\prime}$. By reflecting the interval $(a,b)$ about the midpoint $m_0$ if necessary, we may assume $\tilde{m} > m$. Let $a^\prime = m$ and $b^\prime = \tilde{m}$. Since $L^\prime$ is disjoint from $L$ and $m,\tilde{m}\in\mathcal{C}_{L^\prime}$, we have $a < a^\prime < b^\prime < b$. Since $m > m_0 - (b-a)/4$ and $m<\tilde{m} < b$  we also have
\begin{align*}
b^\prime - a^\prime = \tilde{m} - m < b - \bigg(m_0 - \frac{b-a}{4}\bigg) = \frac{3}{4}(b-a).
\end{align*}
\end{proof}

We may now show that $\mathcal{C}_L$ is always a singleton.
\begin{lemma}\label{singleton_lemma}
If $(U,\varphi)$ is an $\mathcal{F}$-adapted chart and $L\in\mathcal{F}_U$ then $\mathcal{C}_L$ is a singleton.
\end{lemma}
\begin{proof}
Suppose not. Then there is a leaf $L_0\in\mathcal{F}_U$ and $a_0,b_0\in\mathcal{C}_{L_0}$ with $a_0 < b_0$. Applying Lemma \ref{nesting_lemma}, there is a second leaf $L_1$, distinct from $L_0$, and $a_1,b_1\in\mathcal{C}_{L_1}$ with $a_0<a_1<b_1<b_0$ and $b_1 - a_1 < 3 (b_0 - a_0) / 4$. Continuing inductively, there are sequences $\{a_n\}_{n\in\mathbb{N}}$, $\{b_n\}_{n\in\mathbb{N}}$, and a sequence of leaves $\{L_n\}_{n\in\mathbb{N}}$ such that for all $n\in\mathbb{N}$
\begin{gather*}
a_n,b_n\in\mathcal{C}_{L_n}\\
a_n < a_{n+1} < b_{n+1} < b_n\\
b_n - a_n < \bigg(\frac{3}{4}\bigg)^{n-1}(b_{0} - a_{0}).
\end{gather*}
Thus, there is some $\ell\in I$ such that $\lim_{n\rightarrow\infty}a_n = \lim_{n\rightarrow\infty}b_n = \ell$. Let $L^*$ denote the unique leaf in $\mathcal{F}_U$ with $\ell \in\mathcal{C}_{L^*}$.

The set $\mathcal{C}_{L^*}$ is discrete by Lemma \ref{discrete_lemma}. Since $a_n\rightarrow \ell$ and $b_n\rightarrow \ell$, there must therefore be some $n_0\in\mathbb{N}$ such that $a_n,b_n\not\in\mathcal{C}_{L^*}$ for all $n > n_0$. This implies, in particular, disjointness of $L_n$ and $L^*$, for all $n > n_0$. As such, Lemma \ref{nesting_lemma} implies there is an $\ell_n\in (a_n,b_n)$, distinct from $\ell$, with $\ell_n\in\mathcal{C}_{L^*}$, for each $n > n_0$. This contradicts discreteness of $\mathcal{C}_{L^*}$ because $\lim_{n\rightarrow\infty}\ell_n = \ell$.

\end{proof}

With the preceding theoretical framework in place, we can now prove that $h:\pi(U)\rightarrow\mathbb{R}$ is a well-defined homeomorphism onto its image. This result establishes that the leaf space $\mathcal{F}$ has the structure of a (possibly non-Hausdorff) connected, compact, topological $1$-manifold with boundary. 
\begin{lemma}\label{homeo_lemma}
Let $(U,\varphi)$ be an $\mathcal{F}$-adapted chart. There is a well-defined homeomorphism $h : \pi(U)\rightarrow I$ (c.f. Definition \ref{adapted_def}) such that, for each $L\in\mathcal{F}_U$, $\mathcal{C}_L = \{h(L)\}$.  Moreover, the function $H : \pi^{-1}(\pi(U))\rightarrow I$ given by $H = h\circ\pi$ is smooth, has no critical points, and is constant along each $L\in \pi(U)$.
\end{lemma}
\begin{proof}
By Lemma \ref{singleton_lemma}, given $L\in\mathcal{F}_U$ the set $\mathcal{C}_L\subset I$, defined such that $\varphi(U\cap L) = (-\epsilon,\epsilon)^n\times\mathcal{C}_L$, is a singleton. There is therefore a unique function $h:\pi(U)\rightarrow I$ such that $\mathcal{C}_L = \{h(L)\}$ for each $L\in\mathcal{F}_U$. Pulling back $h:\pi(U)\rightarrow I$ along $\pi\mid\pi^{-1}(\pi(U)):\pi^{-1}(\pi(U))\rightarrow \pi(U)$ defines an associated function $H:\pi^{-1}(\pi(U))\rightarrow I:x\mapsto h(\pi(x))$ on $\pi^{-1}(\pi(U))\subset M$. 

It is clear that $H$ is constant along each leaf $L\in\mathcal{F}_U$. Observe that the $(n+1)^{\text{st}}$ component of the diffeomorphism $\varphi:U\rightarrow (-\epsilon,\epsilon)^n\times I$ is related to $H$ according to $\varphi^{n+1}(x) =  h(\pi(x)) = H(x).$
This shows $H\mid U = \varphi^{n+1}$. For $y_0\in\pi^{-1}(\pi(U))$, there is a leaf-preserving diffeomorphism $F:M\rightarrow M$ such that $F^{-1}(U)$ contains $y_0$. So for all $y\in F^{-1}(U)$ we have $H(y) = H(F(y)) = \varphi^{n+1}(F(y))$. This shows that $H:\pi^{-1}(\pi(U))\rightarrow I$ is smooth. Since $F$ is a diffeomorphism and $\varphi^{n+1}$ is a submersion, it also shows that $H$ has no critical points.

We claim the function $h:\pi(U)\rightarrow I$ is continuous. Let $J\subset I$ be open. The preimage $\pi^{-1}(h^{-1}(J)) = H^{-1}(J)$ is open in $\pi^{-1}(\pi(U))$, and therefore $M$, because $H$ is continuous. By the definition of the quotient topology on $\mathcal{F}$, this means the set $h^{-1}(J)$ is open in $\mathcal{F}$. So $h$ is continuous, as claimed. 

We claim $h:\pi(U)\rightarrow I$ is also an open map. Let $V\subset \pi(U)$ be open in $\pi(U)$. Since $\pi$ is open by Lemma \ref{open_map_lemma}, $V\subset \mathcal{F}$ is the intersection of an open set in $\mathcal{F}$ with the open set $\pi(U)\subset\mathcal{F}$. Thus $V$ is open in $\mathcal{F}$ and $\pi^{-1}(V)$ is open in $M$. Since $\varphi^{n+1}$ is an open map, it follows that $\varphi^{n+1}(U\cap\pi^{-1}(V))\subset I$ is open in $I$. But since $H\mid U = \varphi^{n+1}$,
\begin{align*}
\varphi^{n+1}(U\cap\pi^{-1}(V)) = H(U\cap\pi^{-1}(V)) = h(\pi(U\cap\pi^{-1}(V))) = h(\pi(\pi^{-1}(V))) = h(V).
\end{align*}
Note we have used the fact that every leaf in $\pi^{-1}(V)$ meets $U$. It follows that $h$ is open, as claimed.

Finally, we claim $h:\pi(U)\rightarrow I$ is bijective, which shows $h$ is a homeomorphism since every continuous open bijection is a homeomorphism. For surjectivity, we simply note that $h(\pi(U)) =H(U) = \varphi^{n+1}(U) = I $. For injectivity, suppose $L_1,L_2\in\pi(U)$ and $h(L_1) = h(L_2)$. Then $\varphi(L_1\cap U) = \varphi(L_2\cap U)$. But since $L_i$ is the maximal integral manifold that contains $\varphi(L_i\cap U)$, this implies $L_1 = L_2$. Thus, $h$ is bijective, as claimed.
\end{proof}

We may now establish our central claim by showing that $\mathcal{F}$ is in fact a Hasudorff, second-countable, smooth $1$-manifold with boundary.
\begin{lemma}
The leaf space $\mathcal{F}$ is Hausdorff.
\end{lemma}
\begin{proof}
We first claim the following. Let $L\in\mathcal{F}$ and $U$ an open neighborhood of $L$ in $M$. There exists an open neighborhood $\mathcal{U}$ of $L$ in $\mathcal{F}$ such that $\pi^{-1}(\mathcal{U})\subset U$.

For $x\in L$ there is an $\mathcal{F}$-adapted chart $(U_0,\varphi)$ containing $x$ such that $U_0\subset U$. By Lemma \ref{homeo_lemma}, there is a smooth function $H:\pi^{-1}(\pi(U_0))\rightarrow \mathbb{R}$ that is constant along leaves and that has no critical points. Now consider the open neighborhood of $L$ given by $V = U\cap \pi^{-1}(\pi(U_0))$. Fix a Riemannian metric $g$ on $V$ and consider the vector field $N = \nabla H/\| \nabla H\|^2$. By the Flowout Theorem and the Boundary Flowout Theorem [\cite{Lee_2003}, Thm 9.20 and Thm 9.24] and the compactness of $L$, there is a $\delta >0$ and an embedding $\Sigma:L\times I\rightarrow V$, where $I = [0,\delta)$ if $x\in\partial M$ and $I = (-\delta,\delta)$ otherwise, such that for each $y\in L$ the curve $\Sigma_y:[0,\delta)\rightarrow V:t\mapsto \Sigma(y,t)$ is an $N$-integral curve. Since $dH(N) = 1$, we have $(H\circ\Sigma_y)(t) = H(x) + t$. Since $L$ is compact and connected, this implies the images $\Sigma(L\times\{c\})$ for $c\in I$ are leaves of $\mathcal{F}$. The open subset $\Sigma(L\times I)$ is therefore a union of leaves contained in $V$, and thus $U$. That is, $\Sigma(L\times I) = \pi^{-1}(\mathcal{U})$ for some open neighborhood $\mathcal{U}$ of $L$ in $\mathcal{F}$ with $\pi^{-1}(\mathcal{U})\subset U$, as claimed.

Now suppose $L_1,L_2\in\mathcal{F}$ are distinct leaves. Then $L_1\cap L_2 = \emptyset$. By leaf compactness, there are disjoint open sets $U_1,U_2\subset M$ containing $L_1,L_2$, respectively. Applying the above claim, for each $i\in\{1,2\}$, there is an open neighborhood $\mathcal{U}_i$ of $L_i$ in $\mathcal{F}$ such that $\pi^{-1}(\mathcal{U}_i)\subset U_i$. Therefore $\pi^{-1}(\mathcal{U}_1)\cap\pi^{-1}(\mathcal{U}_2) = \emptyset$. Because $\pi$ is surjective it follows that $\mathcal{U}_1\cap\mathcal{U}_2 = \emptyset$. Thus, $\mathcal{F}$ is Hausdorff, as claimed.
\end{proof}

\begin{lemma}
The leaf space $\mathcal{F}$ is second-countable
\end{lemma}
\begin{proof}
For each $L\in\mathcal{F}$ there is an open neighborhood $W$ of $L$ in $\mathcal{F}$, an interval $I\subset \mathbb{R}$, and a homeomorphism $h:W\rightarrow I$. Because $I\subset \mathbb{R}$ and $\mathbb{R}$ is second-countable, $W$ is second countable. Moreover, because $\mathcal{F}$ is compact, it is covered by finitely many such subsets. Thus, $\mathcal{F}$ is second countable.
\end{proof}

\begin{lemma}\label{smooth_structure_lemma}
When equipped with the atlas given by the charts $(\pi(U),h)$ from Lemma \ref{homeo_lemma}, the leaf space $\mathcal{F}$ is smooth $1$-manifold with boundary.
\end{lemma}
\begin{proof}
Since we have already shown that $\mathcal{F}$ is a topological $1$-manifold with boundary, we only need to show smoothness of transition maps.

Let $(U_1,\varphi_1)$, $(U_2,\varphi_2)$, denote $\mathcal{F}$-adapted charts on $M$ with $\pi(U_1)\cap \pi(U_2)\neq \emptyset$. Let $h_i:\pi(U_i)\rightarrow I_i$, $i=1,2$, denote the corresponding homeomorphisms described in Lemma \ref{homeo_lemma}. Define $h_{12}:h_1(\pi(U_1)\cap\pi(U_2))\rightarrow h_2(\pi(U_1)\cap\pi(U_2))$ according to $h_{12} = h_2\circ h_1^{-1}\mid h_1(\pi(U_1)\cap\pi(U_2))$. For each $t_0\in h_1(\pi(U_1)\cap\pi(U_2))$ we will find an open neighborhood $J_0 \ni t_0$  and a formula for $h_{12}\mid J_0$ that is manifestly smooth.

For $t_0\in h_1(\pi(U_1)\cap\pi(U_2))$ there is an $x_0\in U_1$ and a $y_0\in U_2$ with $\pi(x_0) = \pi(y_0)$ and $t_0 = h_1(\pi(x_0)) = H(x_0) = \varphi_1^{n+1}(x_0) $. Choose a leaf-preserving diffeomorphism $F:M\rightarrow M$ with $F(x_0) = y_0$. For $t\in\mathbb{R}$ let $\gamma_0(t) = (\varphi_1^{1}(x_0),\dots,\varphi_1^n(x_0),t)$. Since $x_0 = \varphi_1^{-1}(\gamma_0(t_0))\in U_1$ and $F$ is continuous, there is an open neighborhood $J_0\subset h_1(\pi(U_1)\cap\pi(U_2))$ of $t_0$ such that, for all $t\in J_0$,
\begin{align*}
\varphi_1^{-1}(\gamma_0(t))\in U_1,\quad F(\varphi_1^{-1}(\gamma_0(t)))\in U_2.
\end{align*}
The overlap $h_{12}\mid J_0$ may therefore be computed according to
\begin{align*}
h_2\circ h_1^{-1} &= h_2\circ h_1^{-1}\circ\varphi_1^{n+1}\circ\varphi_1^{-1}\circ\gamma_0\\
& = h_2\circ h_1^{-1}\circ H\circ\varphi_1^{-1}\circ\gamma_0\\
& = h_2\circ h_1^{-1}\circ h_1\circ\pi\circ\varphi_1^{-1}\circ\gamma_0\\
& = h_2\circ \pi\circ\varphi_1^{-1}\circ\gamma_0\\
& = H\circ\varphi_1^{-1}\circ\gamma_0,
\end{align*}
which is the composition of smooth functions.
\end{proof}

\begin{lemma}
With the smooth structure given by Lemma \ref{smooth_structure_lemma}, the map $\pi:M\rightarrow\mathcal{F}$ is a smooth submersion.
\end{lemma}
\begin{proof}
For $x_0\in M$ choose an $\mathcal{F}$-adapted chart $(U,\varphi)$ at $x_0$ and let $h:\pi(U)\rightarrow \mathbb{R}$ denote the corresponding chart on $\mathcal{F}$. Since $h\circ\pi\mid U = H\mid U: U\rightarrow\mathbb{R}$ is a smooth submersion by Lemma \ref{homeo_lemma} and $h$ is a chart on $\mathcal{F}$ it follows that $\pi$ is a smooth submersion at $x_0$.
\end{proof}

\end{document}